\documentclass[sigconf,natbib,screen=true]{acmart}

\PassOptionsToPackage{dvipsnames}{xcolor}

\usepackage{dsfont}
\usepackage{acronym}
\usepackage[noend]{algorithmic}
\usepackage{algorithm}
\usepackage{bbm}
\usepackage{beramono}
\usepackage{bm}
\usepackage{booktabs}
\usepackage[skip=0pt]{caption}
\usepackage{cleveref}
\usepackage[T1]{fontenc}
\usepackage{graphicx}
\usepackage{hyperref}
\usepackage{listings}
\usepackage{multirow}
\usepackage{natbib}
\usepackage{subcaption}
\usepackage{tabularx}
\usepackage[dvipsnames]{xcolor}
\usepackage{enumerate}
\usepackage[inline]{enumitem}
\usepackage{numprint}
\usepackage{mleftright}
\usepackage{lipsum}
\usepackage[toc,page]{appendix}
\usepackage{centernot}
\usepackage{amsthm}

\setcopyright{rightsretained}

\copyrightyear{2022}
\acmYear{2022}
\setcopyright{rightsretained}
\acmConference[ICTIR '22]{Proceedings of the 2022 ACM SIGIR International Conference on the Theory of Information Retrieval}{July 11--12, 2022}{Madrid, Spain}
\acmBooktitle{Proceedings of the 2022 ACM SIGIR International Conference on the Theory of Information Retrieval (ICTIR '22), July 11--12, 2022, Madrid, Spain}\acmDOI{10.1145/3539813.3545142}
\acmISBN{978-1-4503-9412-3/22/07}

\usepackage{etoolbox}
\makeatletter
\patchcmd{\maketitle}{\@copyrightpermission}{
  \begin{minipage}{0.3\columnwidth}
    \href{http://creativecommons.org/licenses/by/4.0/}{\includegraphics[width=0.90\textwidth]{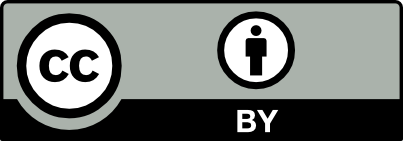}}
  \end{minipage}\hfill
  \begin{minipage}{0.7\columnwidth}
    \href{http://creativecommons.org/licenses/by/4.0/}{This work is licensed under a Creative Commons Attribution International 4.0 License.}
    \end{minipage}
    
   \vspace{5pt}
}{}{}

\definecolor{todo}{RGB}{200, 0, 0}
\definecolor{nk}{RGB}{0, 100, 0}
\definecolor{ho}{RGB}{0, 50, 150}

\acrodef{IR}{Information Retrieval}
\acrodef{LTR}{Learning to Rank}
\acrodef{ARP}{Average Relevance Position}
\acrodef{DCG}{Discounted Cumulative Gain}
\acrodef{EM}{Expectation Maximization}
\acrodef{RecSys}{Recommender Systems}
\acrodef{CTR}{Click-Through Rate}
\acrodef{IPS}{Inverse Propensity Scoring}

\theoremstyle{definition}

\allowdisplaybreaks

\author{Norman Knyazev}
\affiliation{%
	\institution{Radboud University}
	\city{Nijmegen}
	\country{The Netherlands}
}
\email{norman.knyazev@ru.nl}
\author{Harrie Oosterhuis}
\affiliation{%
	\institution{Radboud University}
	\city{Nijmegen}
	\country{The Netherlands}
}
\email{harrie.oosterhuis@ru.nl}

\acmSubmissionID{11}

\title{The Bandwagon Effect: Not Just Another Bias} %

\allowdisplaybreaks
\begin{document}

\begin{abstract}
Optimizing recommender systems based on user interaction data is mainly seen as a problem of dealing with selection bias, where most existing work assumes that interactions from different users are independent.
However, it has been shown that in reality user feedback is often influenced by earlier interactions of other users, e.g.\  via average ratings, number of views or sales per item, etc.
This phenomenon is known as the \emph{bandwagon effect}.

In contrast with previous literature, we argue that the bandwagon effect should not be seen as a problem of statistical bias. In fact, we prove that this effect leaves both individual interactions and their sample mean unbiased.
Nevertheless, we show that it can make estimators inconsistent, introducing a distinct set of problems for convergence in relevance estimation.
Our theoretical analysis investigates the conditions under which the bandwagon effect poses a consistency problem and explores several approaches for mitigating these issues.
This work aims to show that the bandwagon effect poses an underinvestigated open problem that is fundamentally distinct from the well-studied selection bias in recommendation.
\end{abstract}



\maketitle

\acresetall

\begin{figure*}[tb]
\centering
\includegraphics[width=1.72\columnwidth]{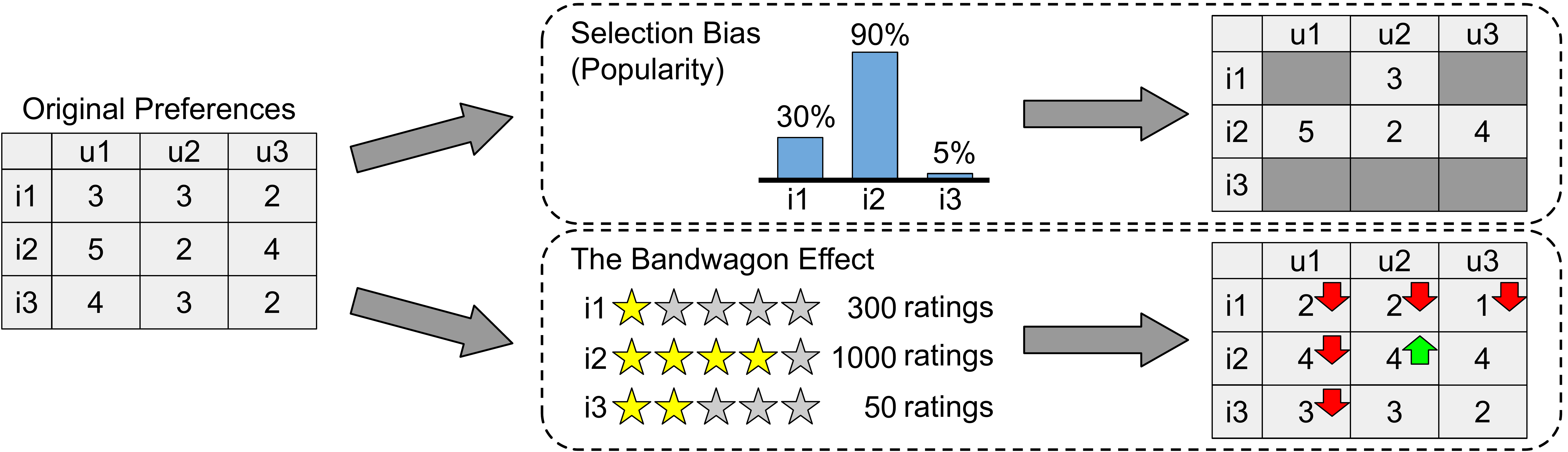}

\vspace{0.25\baselineskip}
\caption{
Impact of selection bias and bandwagon effect on average rating estimation.
Under popularity bias, users are less likely to interact with unpopular items, leaving most of their ratings unobserved.
Under the bandwagon effect, new ratings are skewed in the direction of previous users' opinion, with larger group size eliciting a stronger effect.
}
\label{fig:selection_bias_bandwagon_example} %
\end{figure*}

\section{Introduction}

One of the main challenges in optimizing \ac{RecSys} is correctly capturing user preferences from historical interaction data ~\citep{ricci2015recommender, hu2008collaborative, su2009survey,ilievski2013personalized, jannach2018recommending}.
A widely adopted approach is to train a model to predict known ratings and then apply it to previously unobserved user-item pairs ~\citep{pradel2012ranking, steck2013evaluation, tang2015user}. 
However, in the last decade it has become apparent that observed interactions are often a very skewed representation of the actual user preferences~\citep{steck2011item, pradel2012ranking, canamares2018should}.
In particular, ratings are known to be affected by various forms of \emph{selection} bias, such as \emph{popularity} bias \citep{steck2011item}, \emph{positivity} bias \citep{pradel2012ranking} or \emph{position} bias \citep{oosterhuis2020topkrankings}. 
Popularity bias, for instance, occurs because as most interactions involve only a small number of very popular items, the patterns in observed ratings may not necessarily extrapolate to the data as a whole~\citep{canamares2018should, steck2011item, pradel2012ranking, Schnabel2016}.
It is now well understood that ignoring statistical biases during optimization is likely to lead to suboptimal recommendation performance~\citep{huang2020keeping, Schnabel2016, canamares2018should} 
and numerous solutions have been proposed to address various forms of selection bias in recommendation data~\citep{Schnabel2016, huang2020keeping, saito2020unbiased,wang2019doubly, saito2020doubly, saito2020unbiased, yuan2020unbiased,satoCausalityAwareNeighborhoodMethods2021,jeunen2021pessimistic}.

In this work, we will investigate the related but fundamentally distinct \emph{bandwagon effect},
also referred to as \emph{herding} or \emph{conformity bias}.
Despite being observed by multiple studies, this often-overlooked effect occurs because users tend to follow the behavior and opinions of others~\citep{simon1954Bandwagon,flanagin2013Trusting}.
That is, if a user is told that many other users have liked an item, they themselves become more likely to give a positive rating. %
Conversely, when a user is told that many other users dislike an item, their own rating is likely to become lower. 
Several studies have observed this behavior when exposing the average rating from previous users~\citep{flanagin2013Trusting,sundar2008Bandwagon,cosley2003Seeing,adomavicius2016Understanding} and analogous trends have also been reported for implicit feedback %
~\citep{qin2021Bootstrapping,hanson1996Hits}.
Consequently, the bandwagon effect can pose a serious problem when trying to infer user preferences from interaction data.

Despite its well-known existence, there have only been limited attempts at reducing the influence of the bandwagon effect~\citep{wang2014Quantifying,zhang2017Modeling,liu2016Are} and few authors have theoretically analyzed it in statistical terms~\citep{xie2020Robust,xie2021Understanding}. 
Our work will build on a mathematical model proposed by Xie \textit{et al.}~\citep{xie2020Robust,xie2021Understanding}, who in line with earlier literature~\citep{wang2014Quantifying,zhang2017Modeling} interpret the bandwagon effect as a problem of statistical bias. %

In this work, we provide a rigorous theoretical analysis of the bandwagon effect and its statistical impact on rating 
and \ac{CTR} %
estimation.
In stark contrast with previous work, we prove that a simple form of bandwagon effect of \citeauthor{xie2021Understanding} does \emph{not} introduce statistical bias, as both individual ratings and their sample mean yield unbiased estimates under the effect.
Nonetheless, our analysis also reveals that 
the bandwagon effect can create serious problems for convergence by making the sample mean inconsistent.
Therefore, despite not being a bias issue, the bandwagon effect can still introduce significant errors during optimization. 
To understand when this may pose a problem, we theoretically prove a condition under which inconsistency is guaranteed to arise for the sample mean.
Furthermore, our theoretical and simulated results show that even when an estimator is consistent, the convergence can still be so slow that finding the true user preferences becomes practically infeasible.
We therefore investigate whether existing debiasing methods can be adapted to handle these convergence issues.
However, as none of our methods are able to combine theoretical guarantees with
robust empirical performance, we conclude that the bandwagon 
effect currently remains an open problem that needs more attention in future work.

\section{Related Work}

\subsection{Debiasing Interactions for Recommendation}

Learning from user interactions has a long tradition in the field of recommender systems~\citep{koren2009matrix, pradel2012ranking, Schnabel2016, hernandez2014probabilistic, marlin2009collaborative}.
A common approach is to cast the recommendation problem as a rating prediction task where the goal is to predict the ratings a user would give to items~\citep{koren2009matrix, marlin2009collaborative, hernandez2014probabilistic}.
However, one should not straightforwardly base the optimization on the ratings observed in interaction logs as they are often not representative of actual user preferences due to selection bias~\citep{hernandez2014probabilistic, marlin2009collaborative}.
The most well-known and well-studied form of bias in rating data is popularity bias: a small selection of popular items receives the large majority of user ratings and is thus grossly overrepresented~\citep{canamares2018should, pradel2012ranking, steck2011item}.
Positivity bias is also prevalent: users may be more likely to leave a rating on items they like; consequently, low ratings become underrepresented in the resulting interaction data~\citep{pradel2012ranking}.
As seen in Figure~\ref{fig:selection_bias_bandwagon_example} (top), 
if popular items that attract the majority of ratings also happen to be more liked by users (i.e.\ receive higher ratings), the average rating estimate from observed ratings becomes skewed towards a higher value.
\citet{Schnabel2016} showed that not accounting for selection bias makes the optimization procedure focus on statistically overrepresented user-item combinations.
As a solution, they propose to use \acf{IPS} to unbiasedly estimate and optimize the performance of a system~\citep{little2019statistical}.
Debiasing for recommendation has become a very popular topic in recent years:
for example,
\citet{wang2019doubly} propose a doubly robust estimator that combines \ac{IPS} with regression estimates;
\citet{saito2020unbiased} adapted the \ac{IPS} matrix factorization approach to debias click data instead of ratings;
and
\citet{jeunen2021pessimistic} incorporate a prior distribution to avoid overestimation during prediction.

Selection bias is also widespread in clicks. Its most well-studied form is position bias, which occurs when the position at which an item is displayed heavily impacts the attention the item will receive, and thus how many users will click on it~\citep{craswell2008experimental}.
Similar to ratings, \ac{IPS} is a popular method to correct for selection bias in clicks~\citep{joachims2017unbiased, oosterhuis-phd-thesis-2020}.

The problem of selection bias in explicit feedback, i.e.\ ratings, and implicit feedback, e.g.\ clicks, is thus well recognized and has been addressed extensively by the existing literature~\citep{joachims2017unbiased, oosterhuis-phd-thesis-2020, jeunen2021pessimistic, saito2020unbiased, wang2019doubly, Schnabel2016, koren2009matrix, pradel2012ranking, hernandez2014probabilistic, marlin2009collaborative}.
This can lead to the impression that learning from interactions is mainly a selection bias problem. 
Nevertheless, our current work on the bandwagon effect will show that this is definitely not the case.

\subsection{The Bandwagon Effect}
\label{sec:relatedwork:bandwagon}

The \emph{bandwagon effect}, or \emph{herding}, describes a phenomenon whereby individuals tend to exhibit a greater affinity towards an item when they believe other individuals also like it~\citep{simon1954Bandwagon,raafat2009Herding}.
The decisions of earlier individuals can propagate through \emph{information cascades} to individuals also taking the same decision later, leading to snowballing effects, as well documented in social psychology, economics and finance, and, more recently, web information systems~\citep{raafat2009Herding,caparrelli2004Herding,borghesi2010Spatial,cosley2003Seeing}.

Online platforms frequently include various social cues next to items to assist users with their choice, e.g.\ the number and the average value of previous user ratings. %
Numerous studies have shown that users actively leverage such information when making purchasing or consumption decisions~\citep{hanson1996Hits,dholakia2001Coveted,sundar2008Bandwagon,flanagin2013Trusting,qin2021Bootstrapping,salganik2006Experimental,salganik2008Leading,fu2011Aggregate,kim2016We,huang2006Herding,adomavicius2016Understanding}. 
\citet{hanson1996Hits} found that artificially inflating download counts for some of the products on a software download service led to those products receiving more subsequent downloads than their unaltered counterparts. 
Similarly, \citet{qin2021Bootstrapping} report that Chrome Web Store users rarely click on items displayed with a low download count.
\citet{flanagin2013Trusting} noted 
that users exposed to the average rating also rated movies closer to that average.
Analogously, \citet{sundar2008Bandwagon} found that participants reported higher purchasing intent for high-rated items over low-rated items and that this effect arose from an improved perception of quality and value. 
Both studies also reported that the strength of the effect was positively correlated with the item's number of ratings, consistent with earlier works associating larger group size with stronger herding~\citep{latane1981Social,asch1951Effects,huang2006Herding}.
Altogether, these observations suggest that users incorporate the experience of others into their own judgements and the strength of this effect is correlated with the size of the group.

This phenomenon can have a number of negative consequences.
On an individual level, the bandwagon effect can be so strong that users may act against their own interests: \citet{dholakia2001Coveted} highlight that eBay users frequently bid on items that already have many bids, whilst ignoring comparable or superior listings with fewer bids.
However, another consequence of the bandwagon effect -- and the focus of this paper -- is the difficulty it brings to estimating item relevance from user interactions. 
In the pioneering experiment of \citet{salganik2006Experimental}, a large number of sequentially arriving participants browsed a randomly shuffled grid of the same $48$ songs and could interact with songs of their choice.
The participants were randomly assigned to one of nine groups, with users in groups $2-9$ additionally presented with the download count from previous users \emph{from their group} alongside each song.
The authors observed that songs had a higher probability of being downloaded in a group where they already had a high number of downloads.
In contrast, a song shown next to a low download count was less likely to get downloaded compared to the same song in groups where it had a high or no presented download count (group $1$).
Moreover, with the exception of the very best and worst songs, as determined by the control group $1$, the final ranking according to download counts varied wildly across different groups. 
In their follow-up work, \citet{salganik2008Leading} further showed that the bandwagon effect may be so strong that artificially altering the presented download count can make unpopular songs popular and vice-versa, with only the best songs eventually recovering their true relevance ranking.
Taken together, the above results suggest that under the bandwagon effect even a large amount of user feedback is not guaranteed to yield an accurate estimate of true user preference.
This is also demonstrated in Figure~\ref{fig:selection_bias_bandwagon_example} (bottom): users' true preferences are intertwined with the ratings of earlier users and this influence grows as the number of ratings increases. As a consequence, even when assuming no selection bias, true average rating cannot be determined via the sample mean.

Recently, there has been some interest in modelling herding and disentangling the effects of the crowd on the individual.
\citet{wang2014Quantifying} propose an estimator maximizing a surrogate function of the likelihood of the data and demonstrate the effectiveness of their approach for prediction of subsequent ratings. \citet{zhang2017Modeling} show that subsequent rating estimation can be further improved by considering that some users may actually choose to distance their opinion from the crowd. 
\citet{liu2016Are,zhang2021Causal,wei2021ModelAgnostic} all propose various extensions to classical recommender models by using the rating mean and/or count as additional inputs.
Finally, Xie, Zhong and colleagues formulate the bandwagon process as either a rating reweighting or hidden state prediction problem and show theoretically that the bandwagon effect slows down the convergence of the sample mean~\citep{xie2020Robust,xie2021Understanding,zhong2021Quantifying}.

Notably, most of the above works treat the bandwagon effect as a problem of statistical bias. %
Furthermore, to the best of our knowledge, only the works of \citeauthor{xie2021Understanding}~\citep{xie2020Robust,xie2021Understanding} and \citet{zhong2021Quantifying} attempt to establish a link between the bandwagon effect and statistical theory.
Our theoretical analysis will contrast with the existing literature by showing that the bandwagon effect can affect relevance estimation without introducing statistical bias.

\section{Background}

\begin{figure}[tb]
\centering
\begin{tabular}{c l}
\multirow{3}{4mm}{\rotatebox[origin=c]{90}{
Estimated Value
}}
& \includegraphics[scale=0.52]{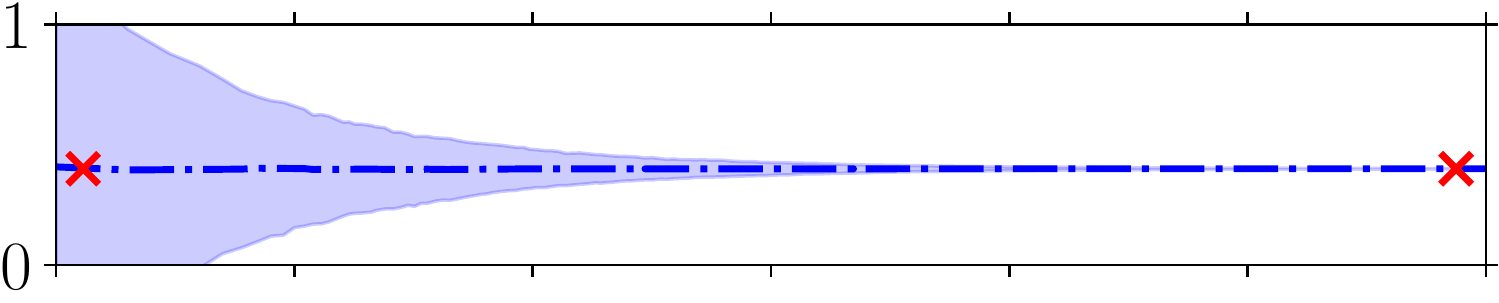}\\
& \includegraphics[scale=0.52]{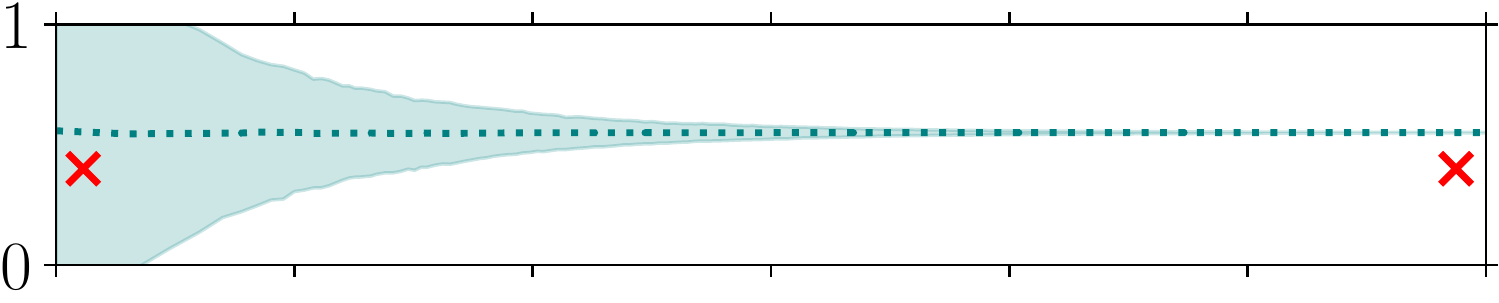}\\
& \includegraphics[scale=0.52]{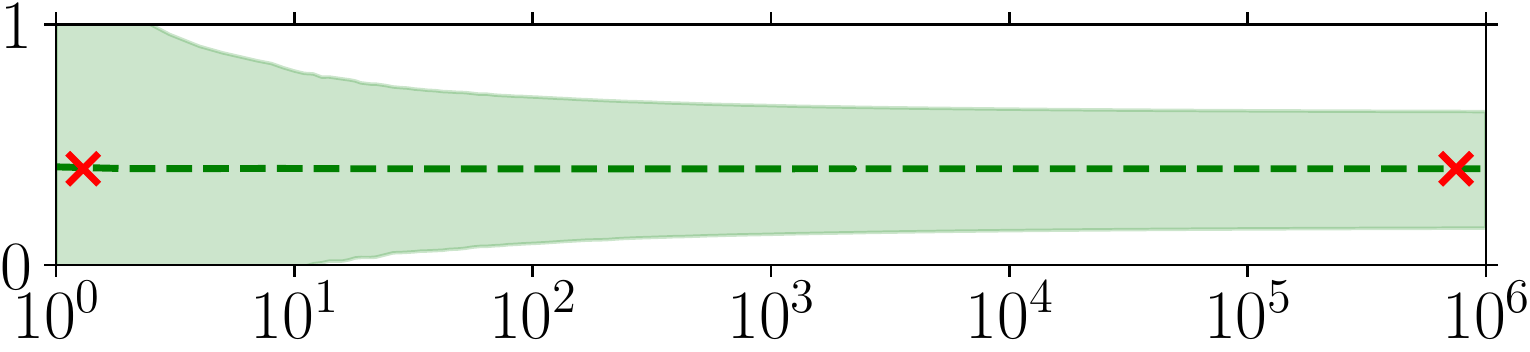}\vspace{0.75mm}
\\
\multicolumn{2}{c}{\hspace{0.5cm}
Number of Samples
}\\
\multicolumn{2}{c}{
\includegraphics[scale=0.54]{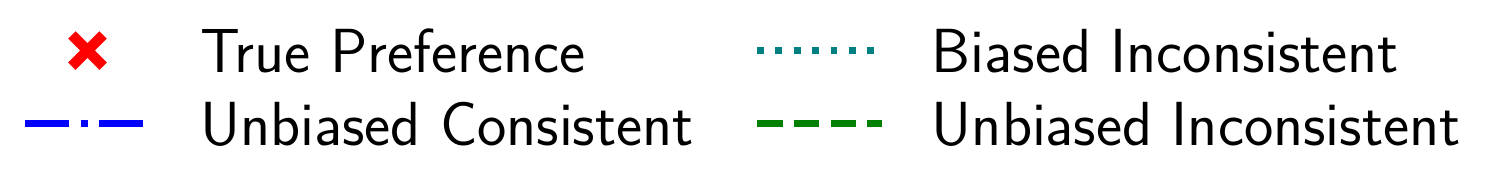}
}
\end{tabular}
\vspace{0.1\baselineskip}
\caption{Three estimators with varying bias and consistency; bold lines show the mean estimate over 1,000 independent runs, shaded areas indicate 90\% confidence intervals.
Top - unbiased and consistent: its expected value and point of convergence are both correct.
Middle - biased and inconsistent: its expected value and point of convergence are both incorrect.
Bottom - unbiased and inconsistent: its expected value is correct but its confidence interval will never converge.%
}
\label{fig:bias_consistency_example} %
\end{figure}

\subsection{Preliminaries}
\label{subsection:preliminaries}

Before we address the bandwagon effect specifically, this section will introduce our problem setting and the statistical properties that our analysis will consider.
For the sake of generality, we focus on \emph{unpersonalized} rating estimation, as we expect the issues observed for this simple task to also apply to the more involved forms of recommendation.
Furthermore, our findings are not specific to recommendation and can be 
readily %
applied to other problems involving relevance estimation such as \ac{CTR} estimation.

In our generic unpersonalized recommendation setting, the goal is to estimate the average preference of all users for each item.
We assume that this value $p \in (0,1)$ exists for each item and represents the percentage of users that truly like the item.
Furthermore, users interact with the items by assigning them binary ratings s.t.\ $r_i \in \{0,1\}$ is the $i$'th rating for an item.
While we refer to $r_i$ as a rating, it can equally well represent many other binary interactions such as \emph{clicks}, \emph{purchases}, \emph{shares}, etc.
As discussed in Section~\ref{sec:relatedwork:bandwagon}, in real systems users are often able to see how many other users have rated an item and what their average rating was, such that past ratings can affect the behaviour of the current user.
We define $p$ to be the \emph{true preference} or the percentage of users who like an item when they have no knowledge of the preferences of any other users.
Our goal is then to estimate $p$ for a \textit{single} item from the given set of observed ratings $\{r_1, r_2, r_3\dots\}$.

Let $\hat{p}_{n}$ indicate an estimate of $p$ based on $n$ observed ratings.
Our analysis will investigate three important beneficial properties that we generally desire from such an estimate.
To start, we normally want the estimate to be unbiased. In other words, in expectation it should be equal to the true value of $p$:
\begin{equation}
\text{Unbiased}\big(\hat{p}_{n}\big) \longleftrightarrow \mathbb{E}\big[\hat{p}_{n}\big] = p.
\end{equation}
Understandably, bias and unbiasedness have received a lot of attention in previous selection bias literature, as it is desirable to avoid systematic errors in rating estimation \citep{Schnabel2016,joachims2017unbiased,oosterhuis2020topkrankings,satoUnbiasedLearningCausal2020,wangCausalInferenceRecommender2020}.
However, unbiasedness on its own does not guarantee that an estimator is actually usable in practice.

The second important and often overlooked property that we may want from an estimator is consistency. An estimator is considered consistent if the estimate converges in probability to the true value of $p$ as the size of available data approaches infinity:
\begin{equation}
\text{Consistent}\big(\hat{p}_{n}\big) \longleftrightarrow \lim_{n \rightarrow \infty}\hat{p}_{n} = p.
\end{equation}
Whilst at a first glance unbiasedness and consistency may appear to refer to the same property, this is certainly \emph{not} the case.
For example, estimators can be biased yet consistent~\citep{swaminathan2015self} or unbiased and inconsistent.
As illustrated in Figure \ref{fig:bias_consistency_example}, the problem with an unbiased inconsistent estimator is that its confidence intervals do not converge to a zero error even as the size of the data increases indefinitely. 
This means its estimation errors will not disappear even after gathering more data, thus preventing accurate convergence.

Nevertheless, whilst consistency guarantees that eventually the estimate will be correct, the amount of data required for convergence may still be unattainable in practice. 
Therefore, the final property we will consider is (data) efficiency, measured by the expected squared error (equivalent to variance for unbiased estimators):
\begin{equation}
\text{Efficiency}\big(\hat{p}_{n}\big) = \mathbb{E}\big[\big(\hat{p}_{n} - p\big)^2\big].
\end{equation}
Overall, efficiency represents a more practical side of an estimator: the error we can expect it to have w.r.t.\ the true value $p$ after $n$ samples.
The effect the number of samples has on this measure can be seen as the speed of convergence.
While most existing work on learning from user interactions focuses on correcting for bias, they often also apply propensity clipping or self-normalization which introduce a little bias to greatly reduce variance~\citep{joachims2017unbiased, satoUnbiasedLearningCausal2020,yangUnbiasedOfflineRecommender2018,saito2020unbiased,swaminathan2015Counterfactuala}.

Our theoretical investigation will consider how the bandwagon effect impacts all three properties: bias, consistency and efficiency.

\subsection{A Formal Model of the Bandwagon Effect}
\label{sec:bandwagonmodel}

The mathematical model of the bandwagon effect we will use in this paper is a simplified version of that of \citeauthor{xie2021Understanding}~\citep{xie2020Robust,xie2021Understanding}.
This model represents a sequential unpersonalized recommendation scenario where users rate the item in succession and one at a time.
It assumes there are two factors that influence the probability of a positive value for the $n$'th rating ($r_n = 1$):
\begin{enumerate*} [label=(\roman*)]
\item $p$, the true proportion of users who would like the item given that they do not know anything of the other users' preferences (no bandwagon effect); and
\item  $\bar{p}_{n-1}$, the mean of the ratings that have taken place previously.
\end{enumerate*}
The rating probability is modelled as a linear combination of these two factors:
\begin{equation}
	\label{eq:process_definition}
		P(r_{n}=1) = \lambda_n p + (1-\lambda_n) \bar{p}_{n-1}, \qquad
		\bar{p}_{n-1} =  \dfrac{1}{n-1}\sum_{i=1}^{n-1} r_i.
\end{equation}
$\lambda_n \in [0,1]$ captures the degree to which users adjust their behavior to match that of the crowd.
For example, $\forall \lambda_n=1$ indicates a process with no bandwagon effect and $\forall n > 1: \lambda_n=0$ a process where true preferences play no role except for the very first user.
As previously mentioned, earlier works have shown a positive correlation with the strength of herding and the size of the group~\citep{latane1981Social,asch1951Effects,huang2006Herding}.
We capture this increase in effect by assuming that the influence of the sample mean cannot decrease as the number of ratings increases:
\begin{equation}
	\label{eq:lambda_definition}
	\lambda_1 = 1 \land \mleft( \forall n>1: 0\leq \lambda_n \leq \lambda_{n-1} \mright).
\end{equation}
In addition, this assumption states that the first rating is completely unaffected by the bandwagon effect (i.e.\ $\lambda_1 = 1$), as the first user has no information about the preferences of others. 
We have chosen the above formulation over others (e.g. \citep{wang2014Quantifying,zhang2017Modeling}) due to its amenability to statistical analysis whilst still capturing the key aspects of the bandwagon effect.%

Before we begin our theoretical analysis, we will briefly discuss two claims from earlier work by \citet{xie2020Robust} so that they can later be compared with our findings in Section~\ref{section:bandwagon_effect}.
First, \citeauthor{xie2020Robust} refer to the bandwagon effect as \emph{persuasion bias}, by which they mean that the bandwagon effect biases the rating distribution such that it no longer represents the true underlying preference.
Second, they show that as long as the lower bound of $\lambda_n$ is above $0$, the bandwagon effect will eventually be eliminated since the sample mean $\bar{p}_{n}$ is guaranteed to converge to the true parameter value:
\begin{equation}
	\inf_{n\in N^+}{\lambda_n}>0 \longrightarrow \lim_{n\to\infty} \bar{p}_n=p.
	\label{eq:earlyconsistency}
\end{equation}
Our theoretical analysis will take another critical look at the bias and consistency of  $\bar{p}_{n}$.
Surprisingly, our findings regarding bias heavily contrast with previous work on the bandwagon effect.

\section{Theoretical Analysis:\\ The Bandwagon Effect}
\label{section:bandwagon_effect}

Now that Section~\ref{sec:bandwagonmodel} has introduced our mathematical model of the bandwagon effect, we will begin our theoretical analysis of its effects on bias, consistency and efficiency.
First, we want to again consider the claim of previous work that the bandwagon effect introduces bias to the rating distribution~\citep{xie2020Robust, xie2021Understanding}.
From our bandwagon model (Equation~\ref{eq:process_definition}) it appears that the probability of a positive rating is indeed a distortion of the true preference $p$.
However, Theorem \ref{theorem:rolling_average_bias} proves that the expected values of all sample means $\bar{p}_n$ and all ratings $r_n$ are actually equal to the true preference:
\begin{equation}
	\forall n: \mathbb{E}\mleft[\bar{p}_n\mright] = \mathbb{E}\mleft[r_n\mright] = p.
\end{equation}
In other words, under our model, both the sample means and actual ratings provide \emph{unbiased} estimates of the true preference, corresponding to either top or bottom, but not center of Figure~\ref{fig:bias_consistency_example}.
This finding contrasts heavily with how previous work has approached the bandwagon effect as a bias problem~\citep{zhang2017Modeling,zhang2020Understanding,xie2020Robust,xie2021Understanding,zhong2021Quantifying,wang2014Quantifying,liu2016Are,wei2021ModelAgnostic} and similarly, how learning from user interactions is generally also framed as a bias problem~\citep{joachims2017unbiased, oosterhuis-phd-thesis-2020, jeunen2021pessimistic, saito2020unbiased, wang2019doubly, Schnabel2016, koren2009matrix, pradel2012ranking, hernandez2014probabilistic, marlin2009collaborative}.
The bandwagon effect \emph{may} be seen as an increased impact of earlier interactions on the estimate of the item's relevance. However, in contrast with selection bias, there is no predetermined rating value that is disproportionally represented for the purposes of average relevance estimation or loss minimization.

Nevertheless, the lack of bias does not mean that the bandwagon effect does not introduce other problems.
As discussed in Section~\ref{subsection:preliminaries}, there are other statistical properties besides unbiasedness that one should consider.
Theorem~\ref{theorem:sample_mean_efficiency} proves the following efficiency of the sample mean, i.e.\ its expected (squared) error:
\begin{equation*}
	\mathbb{E}\mleft[\mleft(\bar{p}_n-p\mright)^2\mright]
	=p\mleft(1-p\mright) \bigg(\dfrac{1}{n^2} + \sum_{i=1}^{n-1} 
	\dfrac{1}{i^2} \prod_{j=i+1}^{n} \dfrac{(j-1)(j+1-2\lambda_j)}{j^2}\bigg).
\end{equation*}
Consistent with our prior discussion, the above formulation reveals that when $\forall \lambda_n = 1$ 
, the variance of $\bar{p}_n$ is equivalent 
to the variance of sampling from the true preference $p$.
Conversely, when $\forall n>1: \lambda_n = 0$, the expected error does not decrease after the first interaction i.e.\ $\forall n: \mathbb{E}[(\bar{p}_n-p)^2] = \mathbb{V}[\bar{p}_1]=p(1-p)$.
Obviously, an error that never decreases provides a serious convergence problem.

To investigate whether the sample mean is consistent, we look at the value of the expected error in the limit of infinite samples.
Theorem~\ref{theorem:sample_mean_asymptotic_efficiency} proves the following expected error in the limit:
\begin{equation*}
		\lim_{n \to \infty} \mathbb{E}\mleft[\mleft(\bar{p}_n-p\mright)^2\mright]
		=
		\hspace{-0.1cm}
		\lim_{n \to \infty} p\mleft(1-p\mright) \bigg(\sum_{i=1}^{n-1}
			\dfrac{1}{i \mleft(i+1\mright)}
			\hspace{-0.1cm}
			\prod_{j=i+1}^{n-1}
				\mleft(1-\dfrac{2\lambda_j}{j+1}\mright)\bigg).
\end{equation*}
Importantly, if the expected error is not zero in the limit as in Figure~\ref{fig:bias_consistency_example} (bottom), the estimator is not consistent and vice-versa.
Based on this fact, we prove a novel necessary condition for the consistency of $\bar{p}_n$ in Theorem~\ref{theorem:sample_mean_consistency}:
		\begin{equation}
			\forall \epsilon: \lim_{n \to \infty} \text{Pr}\mleft(\mleft|\bar{p}_n-p\mright|>\epsilon\mright)=0 \longrightarrow \lim_{n \to \infty} \sum_{i=1}^n \lambda_i^2 = \infty.
		\end{equation}
In other words, if the sample mean is consistent (corresponds to Figure~\ref{fig:bias_consistency_example}, top), the sum of $\lambda_i^2$ terms is unbounded. Equivalently, convergence of $\lim_{n\to\infty}\sum_{i=1}^n \lambda_i^2$ leads to inconsistency.
We note that whilst the sufficient condition of \citet{xie2020Robust} (Equation~\ref{eq:earlyconsistency}) highlights the situations where we can guarantee consistency, it cannot be used to guarantee inconsistency. For instance, one cannot use it to show whether $\lambda_n=n^{-1}$ leads to issues with convergence as $\inf_{n\in N^+}\lambda_n=0$.
On the other hand, as $\lim_{n\to\infty}\sum_{i=1}^n i^{-2}$ is finite, Theorem~\ref{theorem:sample_mean_consistency} shows that $\bar{p}_n$ is inconsistent. Our condition can thus be seen as complementary to the sufficient condition of earlier work, qualifying the set of $\lambda$ sequences where inconsistency is certain.

We can also leverage Theorem~\ref{theorem:sample_mean_asymptotic_efficiency} to quantify the asymptotic expected absolute error for any specific $\lambda$ values.
For instance, Proposition ~\ref{proposition:rolling_average_error} proves that for $\lambda_i=c^{i-1}$ with $c \in (0,1)$:
\begin{align}
&\lim_{n\to\infty}\mathbb{E}\mleft[\mleft|\bar{p}_n-p\mright|\mright]\geq 2p\mleft(1-p\mright) \exp\big( -c\phi \mleft( c,1,2\mright) -c^{2}\phi \big( c^{2},2,2\big) \big) \nonumber \\
&\hspace{4cm}
\text{ s.t. }\phi \mleft(z,s,a\mright) \equiv \sum ^{\infty }_{k=0}\dfrac{z^{k}}{\mleft( a+k\mright)^s }. 
	\label{eq:rolling_average_error_main_text}
\end{align}
Similar bounds may also be derived for other $\lambda$ formulations.

Our analysis proved that neither the sample mean $\bar{p}_n$ nor the individual ratings $r_n$ are biased, a finding that appears to heavily contrast with earlier work that casts the bandwagon effect as a problem of bias~\citep{zhang2017Modeling,zhang2020Understanding,xie2020Robust,xie2021Understanding,zhong2021Quantifying,wang2014Quantifying,liu2016Are,wei2021ModelAgnostic}.
We speculate that the bias framing may at least in part stem from the dependency between ratings. More specifically, conditioning the expectation of $\bar{p}_n$ and $r_n$ on an earlier $\bar{p}_m$ by Theorem~\ref{theorem:rolling_average_bias} and Proposition~\ref{proposition:rolling_average_conditional_bias} leads to:
\begin{align}
	&\forall n > m > 1: 
	\mathbb{E}[\bar{p}_{n-1} \mid \bar{p}_{m}]-p =  \mleft(\bar{p}_m-p\mright)\prod_{i=m+1}^{n-1}\Big(1-\frac{\lambda_i}{i}\Big),\\
	&\forall n > m > 1: \mathbb{E}[r_{n}\mid \bar{p}_{m}]-p=\mleft(1-\lambda_n\mright)\mathbb{E}[\bar{p}_{n-1}-p \mid \bar{p}_{m}].
\end{align}
This can be seen as a \emph{conditional} bias, i.e.\ conditioned on any $\bar{p}_m$ being incorrect, all subsequent ratings and sample means are biased:
\begin{equation}
	\forall n > m \geq 1:  \mathbb{E}[\bar{p}_n \mid \bar{p}_{m} \neq p]\neq p.
\end{equation}
We speculate that previous work has interpreted this conditional bias as bias in the rating distribution and sample mean.
Nonetheless, as discussed above, our bandwagon effect model does not lead to systematic errors in relevance estimation.
In Section~\ref{section:discussion}, we will argue that therefore the bandwagon effect is best interpreted as a problem of consistency and not statistical bias.

To summarize the main findings of our theoretical analysis: 
we have proven that the bandwagon effect does not introduce bias to the rating distribution but can make the sample mean inconsistent.
We have also derived a necessary condition for consistency, expanding on prior work.
Lastly, our formulation of the expected error can quantify the asymptotic error for any specific $\lambda$ values.

\section{Bandwagon Effect Mitigation}
\subsection{Correcting for the Bandwagon Effect}
\label{subsection:correcting_bandwagon_effect}
	
We have established that while the bandwagon effect does not make ratings biased, it can lead to convergence problems.
As possible solutions, inspired by analogous approaches in click debiasing~\citep{vardasbi2020cascade, chuklin2015click}, we will propose an affine and a maximum likelihood estimator.

Firstly, similar to earlier unbiased Learning to Rank work~\citep{vardasbi2020cascade}, we notice that the 
model in Equation~\ref{eq:process_definition} is invertible.
We therefore propose an intermediate affine estimate based on a single rating:
\begin{equation}
		\hat{r}_i = \dfrac{r_i - \mleft(1-\lambda_i\mright)\bar{p}_{i-1}}{\lambda_i}.
	\label{eq:affine_derivation_1_step}
\end{equation}
Theorems~\ref{theorem:affine_single_bias} and~\ref{theorem:affine_single_conditional_bias} prove that our intermediate estimator is both unbiased and conditionally unbiased w.r.t.\ previous ratings:
\begin{equation}
\forall i > j \geq 1: \mathbb{E}[\hat{r}_i \mid \hat{r}_j] = \mathbb{E}[\hat{r}_i \mid \bar{p}_j] = \mathbb{E}[\hat{r}_i ] = p.
	\label{eq:affine_conditional_bias}
\end{equation}
Therefore, using $\hat{r}_i$ instead of $r_i$ in an averaging estimator may allow us to speed up convergence by incorporating only the information about the user's true preference without reinforcing the previous $i-1$ ratings.
With this in mind, we propose our affine estimator that uses all of the observed ratings via a weighted mean:
\begin{equation}
	\label{eq:affine_weighted_mean_definition}
	\hat{p}_n = \dfrac{1}{\sum_{i=1}^n \omega_i}\sum_{i=1}^n \omega_i \hat{r}_i, \quad \text{ s.t.\ } \forall \, 1 \leq i \leq n : \omega_i > 0.
\end{equation}
Clearly $\hat{p}_n$ is also unbiased as it is simply a weighted average over the unbiased estimates $\hat{r}_i$.
Additionally, Theorem ~\ref{theorem:affine_convergence} proves it is consistent if $\sum_{i=1}^n (\omega_i / \lambda_i)^2$ grows slower in $n$ than $(\sum_{i=1}^n\omega_i)^2$ 
, with Corollary~\ref{corollary:affine_convergence} providing specific conditions $\forall\omega_i=1$ and $\forall\omega_i=\lambda_i$.
Furthermore, Theorem~\ref{theorem:affine_general_efficiency} proves $\hat{p}_n$ has the following variance:
\begin{equation}
	\mathbb{V}\mleft[\hat{p}_n\mright]		
	= \dfrac{1}{\mleft(\sum_{i=1}^n \omega_i\mright)^2} \sum_{i=1}^n \dfrac{\omega_i^2}{\lambda_i^2} \mleft( p\mleft(1-p\mright) - \mleft(1-\lambda_i\mright)^2 \mathbb{V}\mleft[\bar{p}_{i-1}\mright]\mright).
	\label{eq:affine_mean_variance_short}
\end{equation}
Unfortunately, interpreting $\mathbb{V}\mleft[\hat{p}_n\mright]$ is not straightforward as in addition to the choice of $\omega_i$ it depends on both the rating values as well as their order.
Analogous to its use in selection bias ~\citep{joachims2017unbiased,strehl2010learning}, one can also apply clipping to the bandwagon effect, e.g. $\hat{\lambda}_i = \max{(\lambda_i, \tau)}$ s.t. $\tau>0$. However, in contrast with \ac{IPS}, $\hat{r}_n$ and $\hat{p}_n$ would remain unbiased, though not conditionally unbiased (Corollary~\ref{corollary:conditional_misestimation_affine_bias}): clipping or misestimation s.t.\ $\hat{\lambda}_i\neq\lambda_i$ may still affect the rate of convergence.

Secondly, inspired by the approach of click models~\citep{chuklin2015click}, we also propose the approximate maximum likelihood estimator $p^*_n$, obtained via iterative application of Newton's method~\citep{galantai2000Theory} to maximize the log-likelihood of the data $L$:
\begin{equation}
	\begin{split}
		L = \sum_{i=1}^n {} & {} r_i \ln\mleft(\lambda_i p^*_n + (1-\lambda_i) \bar{p}_{i-1}\mright) \\[-1.5ex]
		&+ \mleft(1-r_i\mright)\ln\mleft(1- \mleft(\lambda_i p^*_n + (1-\lambda_i) \bar{p}_{i-1}\mright)\mright).
	\end{split}
	\label{eq:maximum_likelihood_derivation}
\end{equation}
Because both the likelihood and the solution found by Newton's method appear intractable, we are unable to prove this estimator is unbiased.
Nonetheless, maximum likelihood estimators have been very effective when applied to user interactions in the past~\citep{zhuang2021cross, chuklin2015click} and can be consistent under certain feasible conditions~\citep{crowder1976Maximum}.

\subsection{Bandwagon Effect Estimation}
	\label{subsection:parameter_estimation}
	\label{sec:effectestimation}

Both our proposed estimators require accurate $\lambda$ values to capture the bandwagon effect well.
Inspired by the experiments of \citet{salganik2006Experimental}, we propose the following parameter estimation scheme that can be applied when dealing with one or more items affected by the same bandwagon effect.

We will assume that $\lambda$ values are shared across a subset of $K\geq1$ items and that all users are randomly assigned to one of $M$ bins s.t. users give $n$ ratings per item and are only presented with rating information regarding users in the same bin.
Let $p^k$ indicate the true preference of all users w.r.t.\ item $k$ (analogous to $p$) and $\hat{\lambda}_i$ an estimate of $\lambda_i$. The log-likelihood of this data is then:
\begin{equation}
	\begin{split}
		L 
		= \sum_{k=1}^K \sum_{m=1}^M\sum_{i=1}^n {}&{} r_{kmi} \ln\big(\hat{\lambda}_i p^{k} + \big(1-\hat{\lambda}_i\big)\bar{p}_{i-1}^{km}\big)
		 \\[-1.75ex]
		&+ \mleft(1-r_{kmi}\mright)\ln\big(1-\hat{\lambda}_{i} p^{k} - \big(1-\hat{\lambda}_i\big)\bar{p}_{i-1}^{km}\big).
	\end{split}
	\label{eq:parameter_estimation_1}
\end{equation}
The maximum likelihood estimate can be found by maximizing $L$, for instance through the iterative application of Newton's method. 
In order to get $\hat{\lambda}_i$ values for $i>n$, we propose fitting them to a curve, e.g.\ $\hat{\lambda}_i = a+(1-a)\cdot b^{i-1}$ where $a \in [0,1]$ and $b \in (0,1)$.
Note that as $p^k$ may be unknown, one may instead estimate $\lambda_i$ values by using its unbiased estimate $\bar{p}^k=\frac{1}{Mn}\sum_{m=1}^M\sum_{i=1}^nr_{kmi}$.

\section{Experiments and Results}

\subsection{Experimental Setup}

To investigate how the bandwagon can affect convergence and whether our proposed estimators can mitigate negative effects, we performed several experiments with simulated user interactions.

For an item with true relevance $p=0.4$, we generated user interactions under three conditions:
a \emph{no bandwagon} setting where $\forall \lambda_i=1$, a \emph{weak bandwagon} setting where $\lambda_i=0.6+0.4\cdot0.9^{i-1}$ and a \emph{strong bandwagon} setting where $\lambda_i=0.1+0.9\cdot0.95^{i-1}$. 
Our comparison includes the following four estimators:
\begin{enumerate*} [label=(\roman*)]
\item the sample mean $\bar{p}_n$ (Equation~\ref{eq:process_definition});
\item affine mean $\hat{p}_n$, which is the affine estimator (Equation~\ref{eq:affine_weighted_mean_definition}) with uniform weights $\forall i: \omega_i = 1$;
\item affine weighted $\tilde{p}_n$, which is the affine estimator but with the weights $\forall i: \omega_i = \lambda_i$; and
\item the maximum likelihood estimate $p_n^*$ (Equation~\ref{eq:maximum_likelihood_derivation}).
\end{enumerate*}
For statistical significance, the entire sampling process was repeated $1{,}000$ times independently, all our graphs include 90\% confidence bounds over these independent runs.

Furthermore, to investigate how sensitive our estimators are to misestimated $\hat{\lambda}$ values, we computed the estimator values with both the actual $\lambda$ values and estimated $\hat{\lambda}$ values.
We followed the estimation scheme of Section~\ref{sec:effectestimation}, generating additional ratings with $K=1$, $p_1=0.4$, $M=20$ and $n=100$ to represent a plausible production-level scenario~\citep{salganik2006Experimental}.
The result obtained from fitting a curve: $\hat{\lambda}_i = a+(1-a)\cdot b^{i-1}$ where $a \in [0,1]$ and $b \in (0,1)$ to the data was: $\hat{\lambda}_i=0.33+0.67\cdot0.92^{i-1}$.
We only report the results for misestimated $\hat{\lambda}$ values in the strong bandwagon setting, due to high similarity with the weak setting.

\subsection{Estimator Bias and Convergence}\label{subsec:estimator_convergence}
\begin{figure}[tb]
\centering
\begin{tabular}{c l}
\rotatebox{90}{\hspace{0.5cm}
Estimated Value
} 
& \hspace{0.1cm}\includegraphics[scale=0.52]{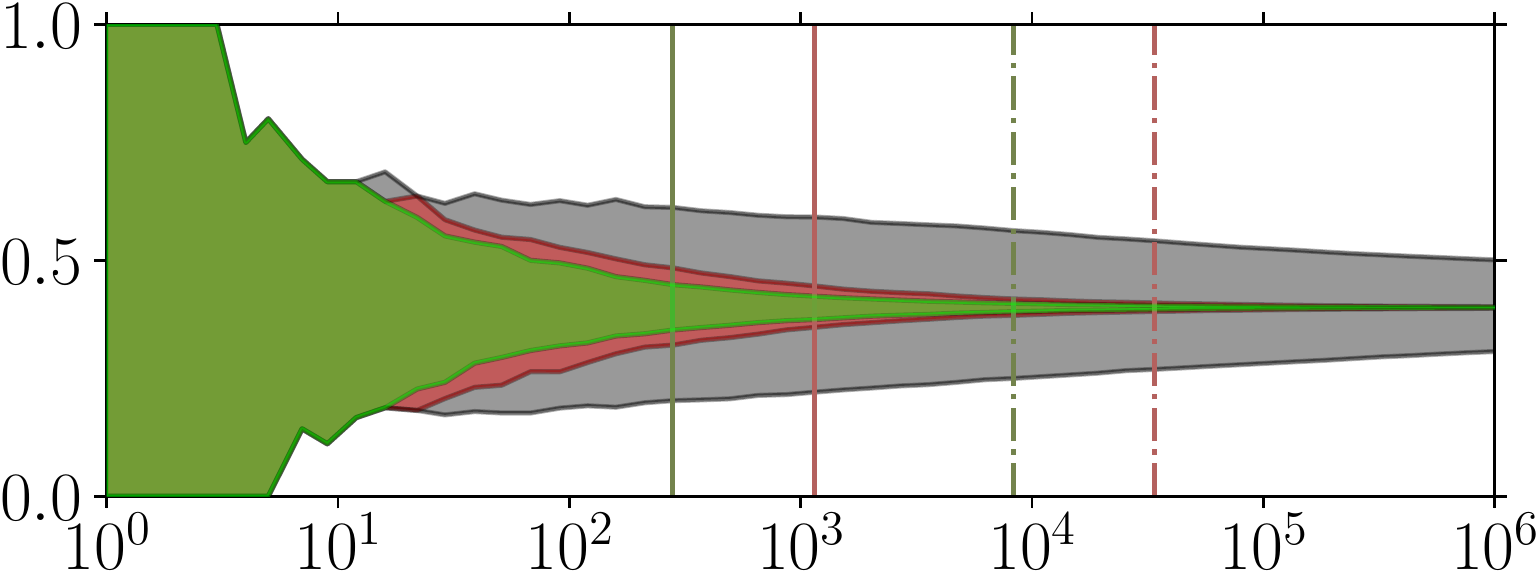}\vspace{0.75mm}\\
\multicolumn{2}{c}{\hspace{0.7cm}
Number of Samples
}\vspace{0.75mm}\\
\multicolumn{2}{c}{
\includegraphics[width=\columnwidth]{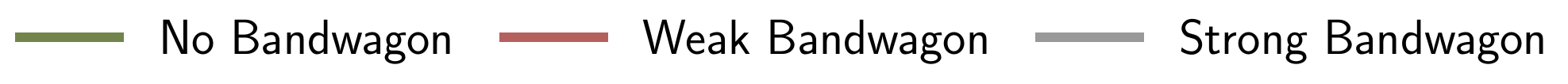}}
\end{tabular}
\vspace{0.1\baselineskip}
\caption{Convergence of $\bar{p}_n$ under various levels of bandwagon. Solid and dashed lines indicate points where a $90\%$ confidence interval over $1{,}000$ independent runs is within $0.05$ and $0.01$ of $p$ respectively, $p=0.4$.}
\label{fig:bw_level_comparison} %
\end{figure}
\begin{figure*}[tb]
\centering
\begin{tabular}{c c c c c}
\rotatebox{90}{\hspace{0.13cm}
\phantom{\emph{Misestimated}}
}
\rotatebox{90}{\hspace{0.13cm}
\emph{Weak Bandwagon}
}
\rotatebox{90}{\hspace{0.28cm}
Estimated Value
}
\raisebox{-0.7mm}{\includegraphics[scale=0.52]{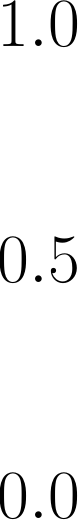}}
\includegraphics[scale=0.52]{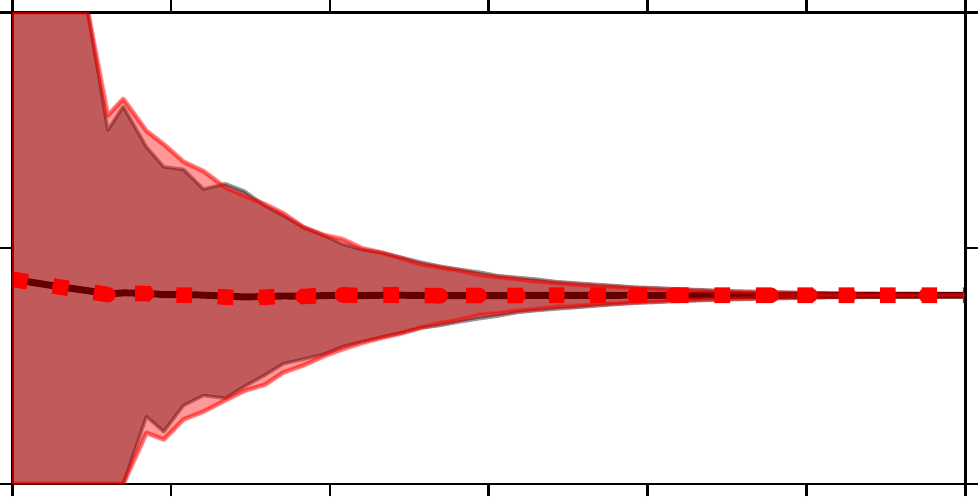} &
\includegraphics[scale=0.52]{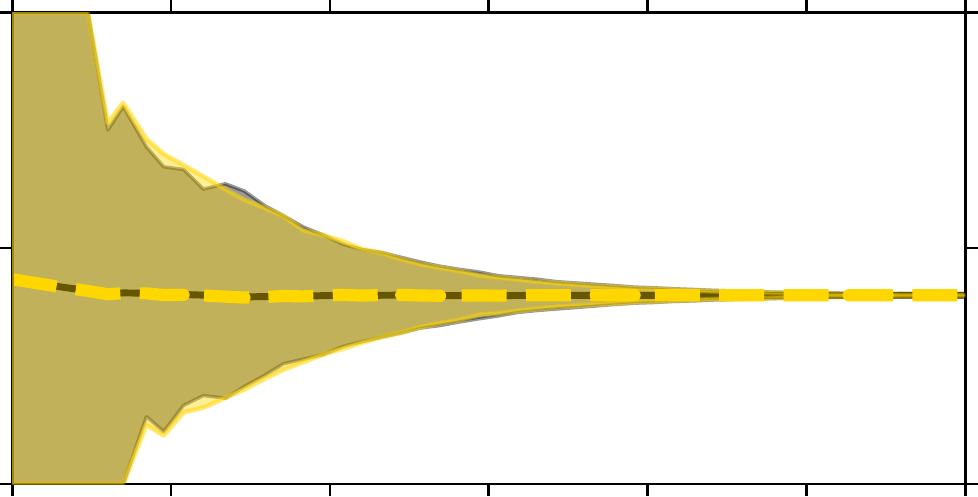} &
\includegraphics[scale=0.52]{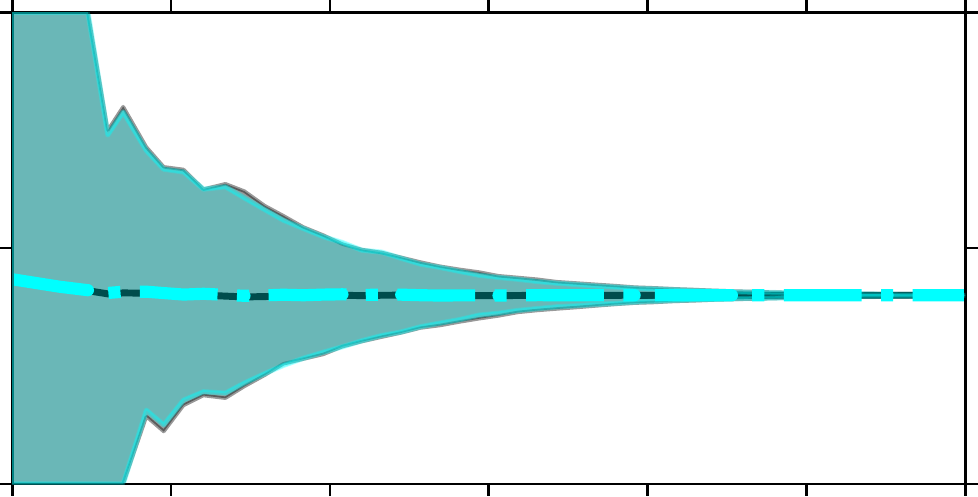}
\vspace{0.5mm}
\\
\rotatebox{90}{\hspace{0.13cm}
\phantom{\emph{Misestimated}}
}
\rotatebox{90}{\hspace{0.13cm}
\emph{Strong Bandwagon}
}
\rotatebox{90}{\hspace{0.28cm}
Estimated Value
}
\raisebox{-0.7mm}{\includegraphics[scale=0.52]{figures/yaxis.pdf}}
\includegraphics[scale=0.52]{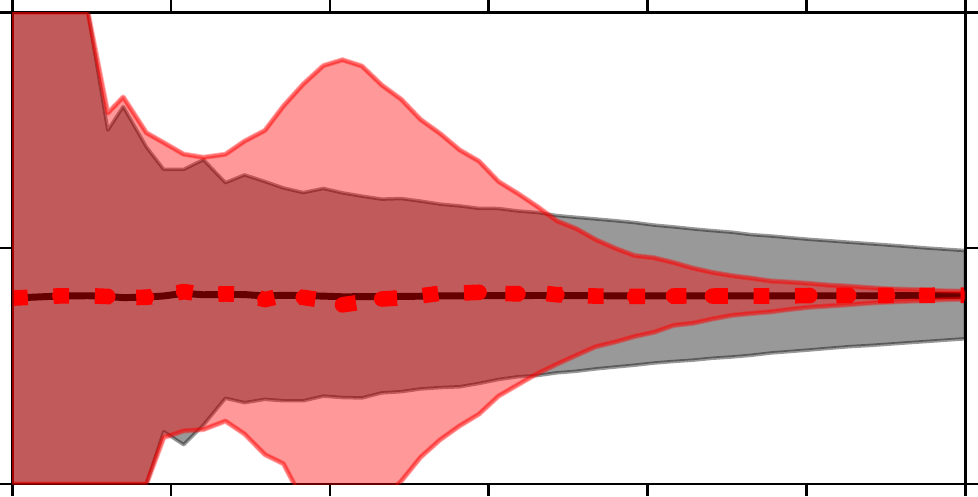} &
\includegraphics[scale=0.52]{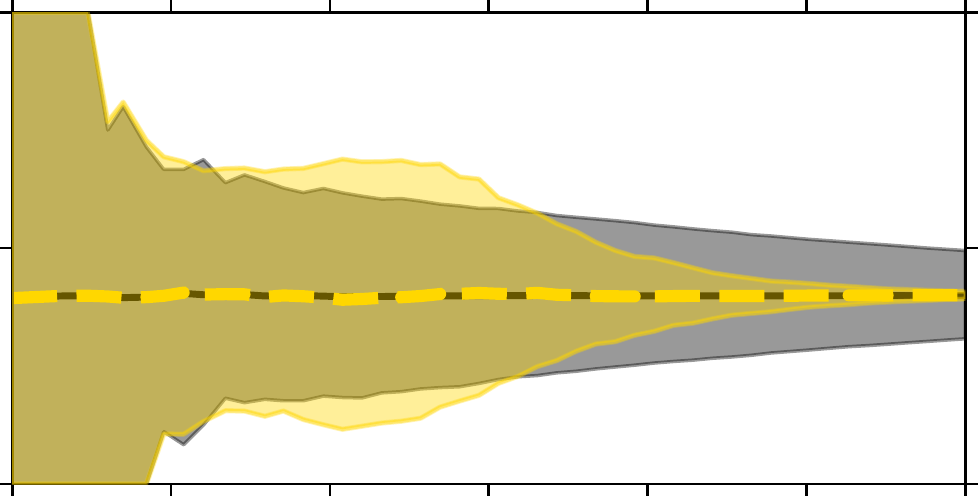} &
\includegraphics[scale=0.52]{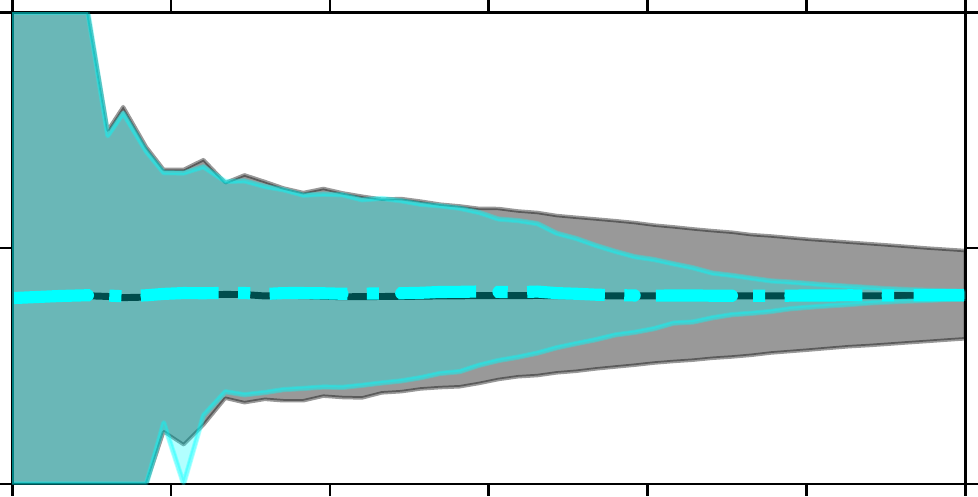}
\vspace{0.5mm}
\\
\rotatebox{90}{\hspace{0.43cm}
\emph{Misestimated} 
}
\rotatebox{90}{\hspace{0.08cm}
\emph{Strong Bandwagon} 
}
\rotatebox{90}{\hspace{0.28cm}
Estimated Value
}
\raisebox{-0.7mm}{\includegraphics[scale=0.52]{figures/yaxis.pdf}}
\includegraphics[scale=0.52]{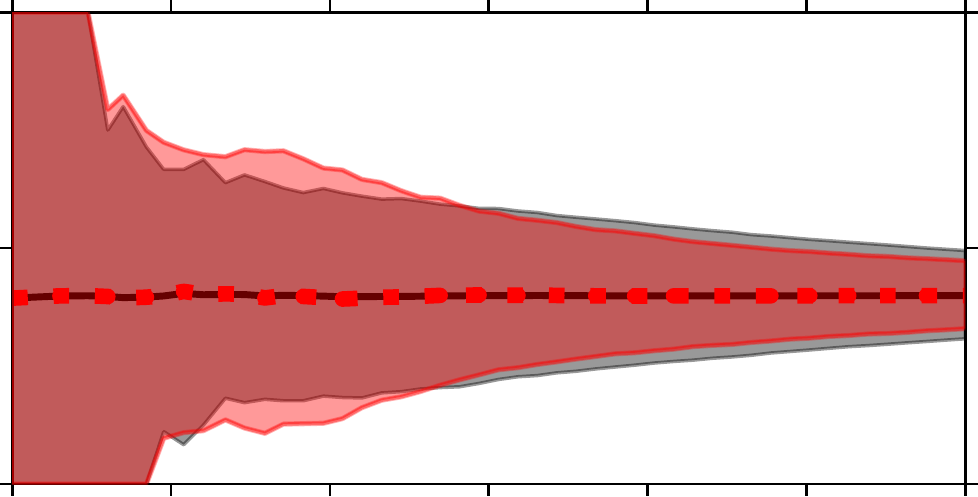} &
\includegraphics[scale=0.52]{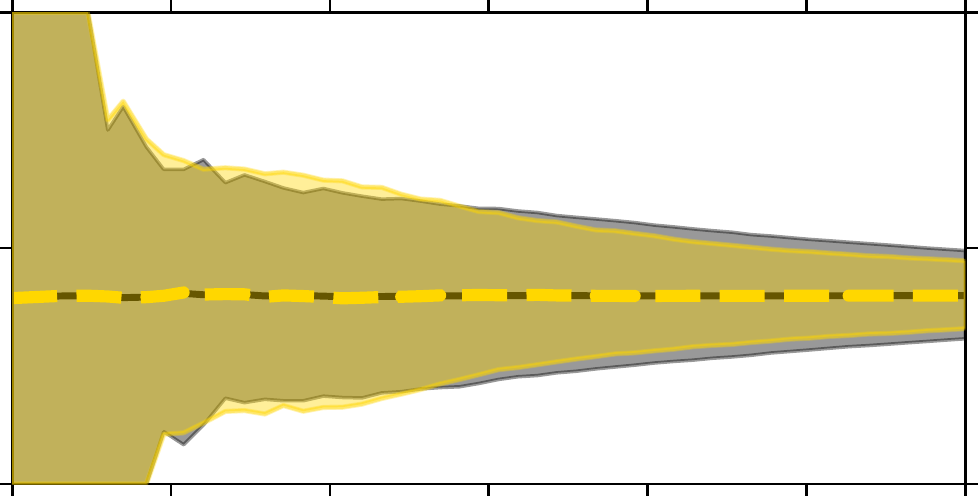} &
\hspace{-0.8mm}\includegraphics[scale=0.52]{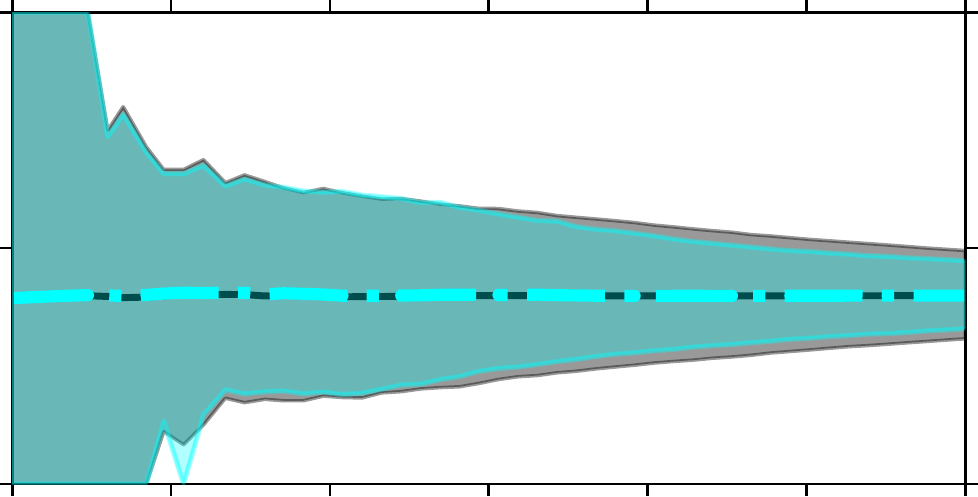}\\
\vspace{1.5mm}
\rotatebox{90}{\hspace{0cm}\phantom{x}}\rotatebox{90}{\hspace{0cm}\phantom{x}}\rotatebox{90}{\hspace{0cm}\phantom{x}}
\hspace{9mm}
\includegraphics[scale=0.52]{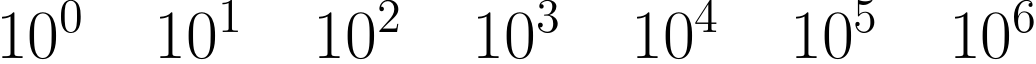} & \hspace{0.1mm}\includegraphics[scale=0.52]{figures/xaxis.pdf} & \hspace{0.1mm}\includegraphics[scale=0.52]{figures/xaxis.pdf}\\
\rotatebox{90}{\hspace{0cm}\phantom{x}}\rotatebox{90}{\hspace{0cm}\phantom{x}}\rotatebox{90}{\hspace{0cm}\phantom{x}}
\vspace{0.75mm}
\hspace{9mm}Number of Samples & Number of Samples & Number of Samples \\
\multicolumn{3}{c}{
\includegraphics[scale=0.5]{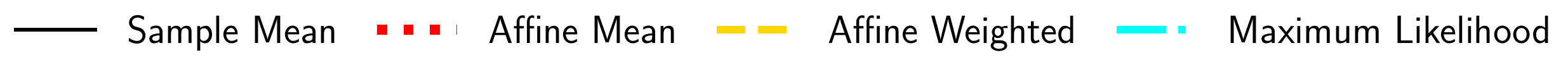}
}\\
\end{tabular}
\vspace{0.1\baselineskip}
\caption{
Convergence of $\hat{p}_n$ (left), $\tilde{p}_n$ (center) and $p_n^*$ (right) compared to $\bar{p}_n$ with $p=0.4$ and $\lambda_i=0.6+0.4\cdot0.9^{i-1}$ (top), $\lambda_i=0.1+0.9\cdot0.95^{i-1}$ (center and bottom) and misestimated $\hat{\lambda}_i=0.33+0.67\cdot0.92^{i-1}$ (bottom); bold lines show the mean estimate over $1{,}000$ independent runs, shaded areas indicate 90\% confidence intervals.
}
\label{fig:estimator_convergence}
\end{figure*}
Our results are visualized in Figure~\ref{fig:bw_level_comparison},
 which displays the convergence rate of the sample mean under various levels of bandwagon effect, and Figure~\ref{fig:estimator_convergence}, which shows the behavior of all estimators with known $\lambda$ in the weak and strong bandwagon settings.

First, we consider the influence of the bandwagon effect on the sample mean ($\bar{p}_n$):
from Figure~\ref{fig:bw_level_comparison} it is immediately clear that the effect's strength has a significant impact on the rate of convergence.
Whilst under no bandwagon effect the confidence intervals are within $0.05$ of $p$ after fewer than $300$ samples and within $0.01$ after roughly $8{,}000$ samples, under the weak bandwagon effect these numbers exceed $1{,}000$ and $33{,}000$ respectively.
Under the strong bandwagon effect the confidence intervals increase even further, exceeding $0.05$ after $10^6$ samples.
However, in all cases the confidence bounds still appear to decrease steadily, suggesting that $\bar{p}_n$ will converge eventually, as also predicted by the sufficient condition for consistency (Equation~\ref{eq:earlyconsistency}).
Nevertheless, it is also clear that the bandwagon effect introduces a significant loss in data efficiency.
In practice, this could mean that an infeasibly large amount of data is required to get a decent estimate under a strong bandwagon effect.

Next, we examine the impact on the bias of the sample mean and other estimators.
As shown in Figure~\ref{fig:estimator_convergence} (top, middle), even prior to convergence, all estimators including $\bar{p}_n$ make negligible errors in expectation, strongly suggesting that they are all unbiased.
Whilst the unbiasedness was already theoretically proven for $\bar{p}_n$ and the affine estimators $\hat{p}_n$ and $\tilde{p}_n$, it is surprising to see that even at low numbers of samples the expected error of the maximum likelihood estimate $p_n^*$ appears negligible.
Our results thus seem to indicate that all the estimators are unbiased or have an extremely small bias, further confirming that the bandwagon effect is not a bias problem.

We then focus on the weak bandwagon condition (Figure~\ref{fig:estimator_convergence} top).
Here, all estimators converge to $p$, at a rate similar to that of $\bar{p}_n$.
That the sample mean and affine estimators converge is not surprising, as the setting meets the sufficient condition for consistency.
However, in this scenario none of the estimators actually provide a clear improvement over the sample mean, indicating that they are unable to apply significant beneficial corrections to the bandwagon effect.
The lack of improvement over $\bar{p}_n$ is particularly concerning as this indicates that our estimators are similarly affected by the loss of data efficiency induced by the bandwagon effect.

Nonetheless, one may expect the convergence to slow down even further in the strong bandwagon setting shown in the middle row of Figure~\ref{fig:estimator_convergence}.
Here, both the affine and maximum likelihood estimators are actually able to converge on the true value much sooner than the sample mean.
However, confidence bounds of the former also become wider after roughly $10$ sampled ratings and stay wider until roughly $1{,}000$ samples, with the weighting of $\tilde{p}_n$ seemingly reducing this effect when compared with $\hat{p}_n$.
Thus during this period, the sample mean actually appears to be a more reliable choice in terms of expected error.
This observation illustrates that choosing between estimators is not always straightforward, as in this case the optimal choice between $\bar{p}_n$ and $\tilde{p}_n$ depends on the number of samples that are available.
Importantly, the maximum likelihood estimator does not have this drawback: its confidence bounds never appear wider than the sample mean while still converging quicker.

Across both settings, we have thus observed no noticeable bias problems for any estimator.
Importantly, the sample mean has serious convergence issues under the strong bandwagon, where the other estimators are able to converge much sooner.
Nonetheless, despite its high efficiency, the affine estimator was also found to initially have the highest variance.
It appears that when $\lambda$ values are known, the maximum likelihood estimator is the most reliable in spite of the lack of proven theoretical guarantees for this estimator.
However, even the maximum likelihood estimator is unable to fully undo the slowdown introduced by either level of the bandwagon effect, highlighting the substantial impact this effect still has on effective sample size and relevance estimation.

\subsection{Bandwagon Misestimation}

To investigate the effect of misestimated $\hat{\lambda}$ values, we look at the bottom row in Figure \ref{fig:estimator_convergence}.
In particular, it shows what happens when we underestimate a strong bandwagon effect.
Clearly, we see that no noticeable bias has been introduced to any estimator.
Interestingly, the affine estimators actually suffer less from initial variance:
however, we speculate that overestimating the bandwagon effect could instead enlarge the variance. %
Nevertheless, all our estimators converge substantially slower than when $\lambda$ was known (cf.\ middle row), with their confidence bounds now only moderately smaller than those of $\bar{p}_n$.
We thus conclude that the ability of our estimators to accelerate convergence is very sensitive to correct estimation of $\lambda$, where we do not observe an advantage for any of the estimators.

The results with misestimated $\hat{\lambda}$ values demonstrate that $\lambda$ estimation is a crucial part of effective bandwagon correction.
Researching bandwagon estimation methods is therefore an important task for future work: since without robust estimation, correcting for the bandwagon effect appears to remain an open problem.

\section{Discussion}
\label{section:discussion}

So far, we have investigated the bandwagon effect in both theoretical and empirical terms.
In contrast with the well-studied selection bias~\citep{Schnabel2016, yuan2020unbiased, satoCausalityAwareNeighborhoodMethods2021, jeunen2021pessimistic}, our experimental results indicate that the bandwagon effect mainly poses a problem for convergence without introducing noticeable statistical bias.
Yet
the effect is generally referred to as a bias by previous work~\citep{xie2020Robust,xie2021Understanding,zhong2021Quantifying,wang2014Quantifying}.
This raises a higher-level conceptual question: \emph{should the bandwagon effect be interpreted as a bias?}
The remainder of this section will contrast three possible perspectives on this question that we deem valuable for the understanding of the field.

On one hand, both our theoretical and empirical results clearly indicate that statistical bias is not the issue.
The expected value of both the sample mean and individual ratings is equal to the true relevance, with
no noticeable error in the mean value of each estimator in our experiments.
Therefore, in statistical terms, it seems that we are purely dealing with convergence and consistency issues and not with bias.
This distinction is important as solutions to inconsistency may be very different than those for bias: i.e.\ our experimental results suggest that methods based on debiasing techniques are not particularly effective at solving convergence issues.

On the other hand, our theoretical analysis also revealed that the bandwagon effect introduces conditional bias:
given that there is an error in the sample mean after $m$ ratings, the expected value of any subsequent rating or sample mean will be incorrect.
In other words, under the bandwagon effect any errors in the sample mean, which are unavoidable in practice, will subsequently introduce (conditional) bias. 
A proponent of the first perspective could argue that this view is mistaking convergence issues for bias.
Nevertheless, we argue that as long as both parties understand the difference between \emph{bias} and \emph{conditional bias}, their disagreement is mainly about semantics and not on the underlying nature of the issue.

Alternatively, one could consider the bandwagon effect as actually changing the \emph{true} preferences of users.
This view poses that the observed ratings are not a distortion of the true preferences but that the actual preferences of users have changed.
One could argue that this is in line with the users' experience as they do not make a conscious choice to alter their ratings in response to others' feedback but aim to assign ratings according to their satisfaction.
Instead, their enjoyment of a product may still be subconsciously influenced by the enjoyment of others.
Thus the advantage of this view is the apparent alignment with the users' own experience. However, their experience is precisely unaware of any subconscious factors such as the bandwagon effect.
The danger is that by not viewing the bandwagon effect as a misrepresentation of preferences, one may be less inclined to correct for it.
However, when training a recommender system, it surely seems critical to know whether high ratings were given because a user was affected by previous users, or whether they made this decision solely based on the item itself.

Apart from these conceptual perspectives, we strongly argue that in statistical terms the bandwagon effect is a problem of convergence and consistency and not of statistical bias.
Nevertheless, outside of pure statistical terms, we think that each of the three views brings a valuable perspective to the field.
Most importantly, we think the field should be aware of the difference between the bandwagon effect and statistical bias.
As long as this distinction is understood, we think it is acceptable if one chooses to conceptually categorize the bandwagon effect as another form of bias, although we would advise against it. %
Overall, our work aims to promote a more complete statistical understanding of the issues surrounding user interactions that goes beyond bias to include variance, consistency and convergence.

\section{Conclusion}

In this paper we examined the bandwagon effect: an underexplored phenomenon where users' interactions are affected by interactions of earlier users. 
While previous work has approached the bandwagon effect as problem of bias~\citep{xie2020Robust,xie2021Understanding}, our theoretical analysis revealed that it does not add bias to the rating distribution nor to the sample mean. 
Instead, we have shown that it is better viewed as a problem of convergence and consistency. %
We derived a necessary condition for consistency under the effect, thereby also adding to earlier consistency claims~\citep{xie2020Robust}.
We then proposed several novel estimators to mitigate the convergence issues.
While our empirical analysis revealed that our novel estimators can substantially increase the convergence rate under a strong bandwagon effect, all estimators also carry their own limitations.
Moreover, even when using the empirically best-performing estimator, our 
purposely simple model of the bandwagon effect still significantly 
increases the number of samples required for accurate relevance estimation.
Future work may wish to examine the impact of more complex bandwagon models (e.g. \citep{zhang2017Modeling,zhong2021Quantifying}) as well as parameter and rating scale choices on bias and convergence, both in unpersonalized as well as \emph{personalized} settings. 
For the latter case in particular, user preference may be further associated with the time of arrival (e.g. item's fans are both more likely to like it and to rate it early) or $\lambda$ values for one timestep may vary across users.
Nevertheless, even under our simple model, our findings lead us to conclude that there is no clear and well-understood solution for the bandwagon effect.

\balance

\subsection*{Acknowledgments}
This work was supported by the Google Research Scholar Program.
All content represents the opinion of the author, which is not necessarily shared or endorsed by their employers and/or sponsors.

\subsection*{Code and data}
To facilitate the reproducibility of the reported results, our experimental implementation is publicly available at \url{https://github.com/NKNY/bandwagonictir2022}.

\appendix
	\section{Appendix}
	\label{appendix:1}
	
\begin{theorem}[Bias of sample mean and single rating]
	\label{theorem:rolling_average_bias}
	$\bar{p}_n$ and $r_n$ (Equation \ref{eq:process_definition}) are unbiased estimators of parameter $p$ for all $n$:
\begin{equation}
\mathbb\forall n: \mathbb{E}\left[\bar{p}_n\right] = \mathbb{E}\left[r_n\right] = p.
\end{equation}
\end{theorem}
\begin{proof}
	\hspace{-0.0mm}$\mathbb{E}[\bar{p}_1]=\mathbb{E}[r_1]=p$ by Equations~\ref{eq:process_definition}-\ref{eq:lambda_definition}. Then by induction:
	\begin{align}
			&\mathbb{E}\mleft[\bar{p}_n\mright]=
			\dfrac{\mleft(n-1\mright)\mathbb{E}\mleft[\bar{p}_{n-1}\mright]+\mathbb{E}\mleft[r_n\mright]}{n}
			=\dfrac{\lambda_n p + \mleft(n-\lambda_n\mright)\mathbb{E}\mleft[\bar{p}_{n-1}\mright]}{n}=p.\nonumber
	\end{align}
	$r_n$ is unbiased as it is a weighted sum of $p$ and unbiased $\mathbb{E}\mleft[\bar{p}_{n-1}\mright]$:
	\begin{align}
		\mathbb{E}\left[r_n\right]
			= \mathbb{E}\left[\lambda_n p + \left(1-\lambda_n\right)\bar{p}_{n-1}\right]
			= \lambda_n p + (1-\lambda_n)\mathbb{E}[\bar{p}_{n-1}]
			= p.\hspace{5mm}\qedhere
	\end{align}
\end{proof}

\begin{lemma}[Variance of $r_i$]
	\label{lemma:variance_of_single_rating}
	Under the bandwagon effect the variance of any rating $r_i$ is $p(1-p)$.
\end{lemma}

\begin{proof}
	By Theorem~\ref{theorem:rolling_average_bias}:
	\begin{align}
			&\forall i: \mathbb{V}\mleft[r_i\mright]
			=\mathbb{E}\mleft[r_i^2\mright]-\mathbb{E}^2\mleft[r_i\mright]
			=\mathbb{E}\mleft[r_i\mright]-\mathbb{E}^2\mleft[r_i\mright]=p-p^2.
			\qedhere
	\end{align}
\end{proof}
	\begin{lemma}[]
	\label{lemma:covariance_of_sample_mean_and_next_rating}
	Covariance of $\bar{p}_{n-1}$ and $r_n$ is $\mleft(1-\lambda_n\mright)^2\mathbb{V}\mleft[\bar{p}_{n-1}\mright]$.
\end{lemma}

\begin{proof}
	\begin{align}
			\mathbb{C}\text{ov} \mleft[\bar{p}_{n-1}, r_n\mright]
			&= \mathbb{E}\mleft[\mleft(\bar{p}_{n-1}-p\mright)\mleft(\lambda_n p + \mleft(1-\lambda_n\mright)\bar{p}_{n-1}-p\mright)\mright] \\
			&= \mleft(1-\lambda_n\mright)^2\mathbb{V}\mleft[\bar{p}_{n-1}\mright]. \qedhere
	\end{align}
\end{proof}
	\begin{theorem}[Efficiency of sample mean]
		\label{theorem:sample_mean_efficiency}
		The efficiency of the sample mean under the bandwagon effect defined in Equation~\ref{eq:process_definition} is \[
		\mathbb{E}\mleft[\mleft(\bar{p}_n-p\mright)^2\mright]
			=p\mleft(1-p\mright) \bigg(\dfrac{1}{n^2} + \sum_{i=1}^{n-1} 
			\dfrac{1}{i^2} \prod_{j=i+1}^{n} \dfrac{(j-1)(j+1-2\lambda_j)}{j^2}\bigg).
		\]
\end{theorem}
\begin{proof}
	By Theorems~\ref{theorem:rolling_average_bias} and Lemmas~~\ref{lemma:variance_of_single_rating}, \ref{lemma:covariance_of_sample_mean_and_next_rating}:
		\begin{align}
			&\mathbb{E}\mleft[\mleft(\bar{p}_n-p\mright)^2\mright]
			= \dfrac{\mleft(n-1\mright)^2}{n^2}\mathbb{V}\mleft[\bar{p}_{n-1}\mright]+\dfrac{2\mleft(n-1\mright)}{n^2}\mathbb{C}\text{ov}\mleft[r_n,\bar{p}_{n-1}\mright] \nonumber\\
			&\hspace{0.55cm}+\dfrac{\mathbb{V}\mleft[r_n\mright]}{n^2} =\dfrac{\mleft(n-1\mright)\mleft(n+1-2\lambda_n\mright)}{n^2}\mathbb{V}\mleft[\bar{p}_{n-1}\mright]+\dfrac{p(1-p)}{n^2}.
		\end{align}	
	The proof is completed by unrolling the process to $\mathbb{V}\mleft[\bar{p}_{1}\mright]=p(1-p)$ and collecting all the terms.\qedhere
\end{proof}
	\begin{theorem}[]%
		\label{theorem:sample_mean_asymptotic_efficiency}
		The asymptotic effiency of the sample mean is:
		\[
			\lim_{n \to \infty} \mathbb{E}\mleft[\mleft(\bar{p}_n-p\mright)^2\mright]
			= \lim_{n \to \infty} p\mleft(1-p\mright) \bigg(\sum_{i=1}^{n-1}
				\dfrac{1}{i \mleft(i+1\mright)}\prod_{j=i+1}^{n-1}
					\bigg(1-\dfrac{2\lambda_j}{j+1}\bigg)\bigg).
		\]
\end{theorem}
\begin{proof}
	By re-grouping numerators and denominators in the product of Theorem~\ref{theorem:sample_mean_efficiency} and cancelling out extra terms:
	\begin{align}
		&\mathbb{E}\mleft[\mleft(\bar{p}_n-p\mright)^2\mright]
		=p\mleft(1-p\mright) \mleft(\dfrac{1}{n^2} + \sum_{i=1}^{n-1}
				\dfrac{n+1-2\lambda_n}{i \mleft(i+1\mright)n}\prod_{j=i+1}^{n-1}
					\mleft(1-\dfrac{2\lambda_j}{j+1}\mright)\,\mright)\nonumber.
	\end{align}
	The proof is completed by taking $\lim_{n\to\infty}$ of above expression.
\end{proof}

	\begin{theorem}[Sample mean consistency]
		\label{theorem:sample_mean_consistency}
		$\sum\lambda_i^2$ diverging is a necessary condition for the consistency of $\bar{p}_n$ (Equation~\ref{eq:process_definition}):
		\begin{equation}
			\forall \epsilon>0: \lim_{n \to \infty} \text{Pr}\mleft(\mleft|\bar{p}_n-p\mright|>\epsilon\mright)=0 \longrightarrow \lim_{n \to \infty} \sum_{i=1}^n \lambda_i^2 = \infty.
		\end{equation}
\end{theorem}
\begin{proof}
  By \citet[p. 159]{loeve1977Probability}, consistency of $\bar{p}$ implies its variance vanishing: $\forall \epsilon>0: \text{Pr}(|\bar{p}_n-p|\geq\epsilon)\geq\mathbb{V}\left[\bar{p}_n\right]-\epsilon^2$.
	Then: 
	$\forall \epsilon>0: \lim_{n \to \infty} \text{Pr}\left(\left|\bar{p}_n-p\right|>\epsilon\right)=0
	\rightarrow \lim_{n \to \infty} \text{Var}(\bar{p}_n)=0
	\rightarrow \forall i: \lim_{n \to \infty} \prod_{j=i+1}^{n-1}(1-\frac{2\lambda_j}{j+1})=0
	\rightarrow \forall i: \lim_{n \to \infty} \sum_{j=i+1}^{n-1}\frac{\lambda_j}{j+1}=\infty \rightarrow \forall i:\lim_{n\to\infty} \sqrt{\sum_{j=i+1}^{n-1}\lambda_j^2}\sqrt{\sum_{j=i+1}^{n-1}\frac{1}{\left(j+1\right)^2}}=\infty \rightarrow \lim_{n \to \infty} \sum_{i=1}^n \lambda_i^2 = \infty$.
	The second step is due to partial products in Theorem \ref{theorem:sample_mean_asymptotic_efficiency}, the third by \citet[p. 6]{leonard2012Notes}, the fourth by Cauchy-Schwarz inequality and the fifth by convergence of p-series with $p=2$.
	\qedhere

\end{proof}
	
\begin{proposition}[Conditional bias of sample mean]
	\label{proposition:rolling_average_conditional_bias}
	\hspace{0.1cm}\\
	$\mathbb{E}\mleft[\bar{p}_n-p\mid\bar{p}_m\mright]=\mleft(\bar{p}_m-p\mright)\prod_{i=m+1}^n\mleft(1-\frac{\lambda_i}{i}\mright)$ where $m<n$.
\end{proposition}
\begin{proof}
	By repeatedly applying the induction step in Theorem ~\ref{theorem:rolling_average_bias}, replacing $\mathbb{E}\mleft[\bar{p}_n\mright]$ with $\mathbb{E}\mleft[\bar{p}_n\mid\bar{p}_m\mright]$ and noting that $\mathbb{E}\mleft[\bar{p}_m|\bar{p}_m\mright]=\bar{p}_m$:
	\begin{equation*}
		\mathbb{E}\mleft[\bar{p}_n-p\,|\,\bar{p}_m\mright]
		= \mathbb{E}\mleft[\bar{p}_{n-1}-p\,|\,\bar{p}_m\mright]\big(1-\frac{\lambda_n}{n}\big)
		= \big(\bar{p}_{m}-p\big)\hspace{-1mm}\prod_{i=m+1}^n\hspace{-1mm}\big(1-\frac{\lambda_i}{i}\big).\qedhere
	\end{equation*}
\end{proof}
	\begin{proposition}[Error bound of $\bar{p}_n$]
	\label{proposition:rolling_average_error}
	If $\lambda_i=c^{i-1}$ s.t. $c \in \mleft(0,1\mright)$, the asymptotic expected absolute error of the sample mean is at least $2p\mleft( 1-p\mright) \exp \mleft( -c\phi\mleft( c,1,2\mright)-c^{2}\phi \mleft( c^{2},2,2\mright) \mright)$, where $\phi( z,s,$ $a) \equiv \sum ^{\infty }_{k=0}\frac{z^{k}}{\mleft( a+k\mright)^s }$ is Lerch transcendent and converges to a fixed value $\forall z<1$ \citep{ferreira2017New}. 
\end{proposition}
\begin{proof}
		\begin{align}
			&\lim_{n \to \infty} \mathbb{E}\mleft[\mleft|\bar{p}_n-p\mright|\mright] 
			\geq \lim_{n \to \infty} \mathbb{E}\mleft[\mleft|\mathbb{E}\mleft[\bar{p}_n-p|r_1\mright]\mright|\mright]
			= 2p\mleft(1-p\mright)\nonumber\\
			&\cdot\exp\mleft(\sum_{i=2}^{\infty}\ln\mleft(1-\dfrac{\lambda_i}{i}\mright)\mright)
			\geq 2p\mleft(1-p\mright)\exp\bigg(\sum_{i=2}^{\infty}\Big(-\dfrac{\lambda_i}{i}-\dfrac{\lambda_i^2}{i^2}\Big)\bigg)\nonumber \\
			&=2p\mleft(1-p\mright) \exp \Bigg( -c\sum ^{\infty}_{i=0}\dfrac{c^{i}}{\mleft( 2+i\mright) }
			-c^{2}\sum ^{\infty}_{i=0}\dfrac{c^{2i}}{\mleft( 2+i\mright)^2 }\Bigg).\nonumber
		\end{align}
		\label{eq:rolling_average_error}
	The first step is due to the law of total expectation and Jensen's inequality for conditional expectation. The second is by Proposition~\ref{proposition:rolling_average_conditional_bias} with $m=1$. The third is due to a known inequality $\forall x> -0.68: \ln\mleft(1+x\mright) \geq x-x^2$ \citep{kozma2021Useful}.
\end{proof}
	\begin{theorem}[]
		\label{theorem:affine_single_bias}
		Each $\hat{r}_{i}$ (and thus also $\hat{p}_n$) is unbiased for all $\hat{\lambda}_i$.
\end{theorem}
\begin{proof}
	By Theorem ~\ref{theorem:rolling_average_bias}, for any choice of $i$ and $\hat{\lambda}_i$:
	\begin{equation*}
			\mathbb{E}\mleft[\hat{r}_{i} \mid \hat{\lambda}_i\mright]
			= \mathbb{E}\bigg[ \dfrac{r_i-(1-\hat{\lambda}_i)\bar{p}_{i-1}}{\hat{\lambda}_i}\bigg] 
			= \dfrac{\mathbb{E}\mleft[ r_i-\bar{p}_{i-1} \mright]}{\hat{\lambda}_i} + \mathbb{E}\mleft[ \bar{p}_{i-1} \mright]
			= p.
\qedhere
	\end{equation*}
\end{proof}
	\begin{theorem}[Conditional bias of $\hat{r}_{i}\mid\hat{\lambda}_i=\lambda_i$]
		\label{theorem:affine_single_conditional_bias}
		If $\hat{\lambda}_i=\lambda_i$, the affine mean estimator $\hat{r}_{i}$ is an unbiased estimator of $p$ when conditioned on $\hat{r}_j$ or $\bar{p}_{j}$ s.t. $\forall j<i$.
\end{theorem}
\begin{proof}
	For any function of previous ratings $f_{i}=f(r_1, \ldots, r_{i-1})$:
	\begin{align}
			\mathbb{E}\left[\hat{r}_{i}\mid f_{i},\hat{\lambda}_i\right] &= \mathbb{E}\bigg[\dfrac{\lambda_i p + (1-\lambda_i)\bar{p}_{i-1}-(1-\hat{\lambda}_i)\bar{p}_{i-1}}{\hat{\lambda}_i}\mid f_{i}\bigg] \nonumber\\
			&=\dfrac{\lambda_i}{\hat{\lambda}_i}p+\mleft(\dfrac{\hat{\lambda}_i-\lambda_i}{\hat{\lambda}_i}\mright)\mathbb{E}\mleft[\bar{p}_{i-1}\mid f_i\mright] = p.
		\label{eq:affine_conditional_bias}
	\end{align}
	The proof is completed by setting $f_i=\hat{r}_j$ or $f_i=\bar{p}_{j}$.
\end{proof}

\begin{corollary}[Conditional bias of $\hat{r}_{i}\mid\hat{\lambda}_i\neq\lambda_i$]
\label{corollary:conditional_misestimation_affine_bias}
		Affine estimator $\hat{r}_{i}$ is not unbiased when conditioned on a function of previous interactions as the last equality of Equation~\ref{eq:affine_conditional_bias} may not hold.
\end{corollary}

	\begin{lemma}[Covariance of affine estimators]
	\label{lemma:affine_covariance}
	$\hat{r}_{i}$ and $\hat{r}_{j<i}$ are uncorrelated as long as $\hat{\lambda}_i=\lambda_i$.
\end{lemma}
\begin{proof}
	By Theorem~\ref{theorem:affine_single_conditional_bias} and the law of total expectation:
	\begin{align}
			\mathbb{E}\mleft[\hat{r}_{i}\hat{r}_{j}\mright]
			&= \mathbb{E}\mleft[\mathbb{E}\mleft[\hat{r}_{i}\hat{r}_{j}\mid \hat{r}_{j}\mright]\mright]
			= \mathbb{E}\mleft[\hat{r}_{j}\mathbb{E}\mleft[\hat{r}_{i}\mid \hat{r}_{j}\mright]\mright] \nonumber \\
			&= \mathbb{E}\mleft[\hat{r}_{j}\mathbb{E}\mleft[\hat{r}_{i}\mright]\mright]
			= \mathbb{E}\mleft[\hat{r}_{j}\mright]\mathbb{E}\mleft[\hat{r}_{i}\mright]. \qedhere
	\end{align}
\end{proof}

\begin{theorem}
\label{theorem:affine_general_efficiency}
The variance of affine estimator in Equation~\ref{eq:affine_weighted_mean_definition} is:
\[
	\mathbb{V}\mleft[\hat{p}_n\mright]		
	= \frac{1}{\mleft(\sum_{i=1}^n \omega_i\mright)^2} \sum_{i=1}^n \frac{\omega_i^2}{\lambda_i^2} \mleft( p\mleft(1-p\mright) - \mleft(1-\lambda_i\mright)^2 \mathbb{V}\mleft[\bar{p}_{i-1}\mright]\mright).
\]
\end{theorem}

\begin{proof}
	By Lemmas~\ref{lemma:affine_covariance},~\ref{lemma:covariance_of_sample_mean_and_next_rating} and ~\ref{lemma:variance_of_single_rating}:
	\begin{align}
			&\mathbb{V}\mleft[\hat{p}_n\mright]
			= \dfrac{1}{\mleft(\sum_{i=1}^n \omega_i\mright)^2} \mleft(\sum_{i=1}^n \omega_i^2\mathbb{V}\mleft[\hat{r}_{i}\mright] + 2\sum_{i=1}^{n-1}\sum_{j=i+1}^n \omega_i\omega_j\mathbb{C}\text{ov}\mleft[\hat{r}_{i}, \hat{r}_{j}\mright]\mright)
			\label{eq:affine_variance} \nonumber\\
			&= \dfrac{1}{\mleft(\sum_{i=1}^n \omega_i\mright)^2} \sum_{i=1}^n \dfrac{\omega_i^2}{\lambda_i^2}\big(\mathbb{V} \mleft[r_i\mright] + \mleft(1-\lambda_i\mright)^2 \mathbb{V}\mleft[\bar{p}_{i-1}\mright] - 2\mleft(1-\lambda_i\mright) \nonumber\\
			&\cdot\mathbb{C}\text{ov}\mleft[r_i, \bar{p}_{i-1}\mright]\big)
			= \mleft(\sum_{i=1}^n \omega_i\mright)^{-2} \hspace{-1.5mm}\sum_{i=1}^n \dfrac{\omega_i^2}{\lambda_i^2} \mleft(p-p^2 - \mleft(1-\lambda_i\mright)^2 \mathbb{V}\mleft[\bar{p}_{i-1}\mright]\mright).\qedhere
	\end{align}
\end{proof}	
	
\begin{theorem}[Convergence of affine estimator]
	\label{theorem:affine_convergence}
	When $\forall\hat{\lambda}_i=\lambda_i$, $\hat{p}_n=\frac{\sum_{i=1}^n \omega_i \hat{r}_i}{\sum_{i=1}^n \omega_i}$ converges within $\epsilon$ of the true value $p$ with probability of at least $1-\alpha$ if $\mleft(\sum_{i=1}^n \omega_i\mright)^2 \geq \frac{1}{2\epsilon^2}\log{\frac{2}{\alpha}}\sum_{i=1}^n \big(\frac{\omega_i}{\lambda_i}\big)^2$.

\end{theorem}

\begin{proof}
Define $X_n$ as $\sum_{i=1}^n \omega_i \mleft(\hat{r}_{i}-p\mright)$, $X_0=0$. Then $X_n$ is a martingale as $\mathbb{E}\mleft[X_n \mid X_0,X_1,...,X_{n-1}\mright]=X_{n-1}+\mathbb{E}\mleft[w_i\mleft(\hat{r}_i-p\mright)\mid\bar{p}_{n-1}\mright]=X_{n-1}$. For each $n: X_n-X_{n-1}=\omega_n \mleft(\hat{r}_{n}-p\mright)$, whose values for $r_n=0$ and $r_n=1$ are separated by $d_n=\frac{\omega_n}{\lambda_n}$. By Azuma's inequality:
\begin{equation}
	\begin{split}
		\text{Pr}\mleft(\mleft|\dfrac{\sum_{i=1}^n \omega_i \mleft(\hat{r}_{i}-p\mright)}{\sum_{i=1}^n \omega_i}\mright| > \epsilon \mright) 
		&=\text{Pr}\bigg(|X_n-X_0| > \epsilon \sum_{i=1}^n \omega_i\bigg) \\
		&\leq 2\exp\mleft(\dfrac{-2\epsilon^2\mleft(\sum_{i=1}^n \omega_i\mright)^2}{\sum_{i=1}^n d_i^2}\mright).
	\end{split}
\end{equation}
The proof is completed by taking the logarithm of both sides and redistributing the terms.
\end{proof}

\begin{corollary}[Convergence of $\bar{p}_n$ for specific $\omega_i$]
	\label{corollary:affine_convergence}
	By Theorem \ref{theorem:affine_convergence}, when $\forall{}i:\hat{\lambda}_i=\lambda_i$, $\omega_i=1$, $\hat{p}_n$ is a consistent estimator of $p$ if $\sum_{i=1}^n\frac{1}{\lambda_i^2}$ grows slower than $n^2$. Similarly, when $\forall\hat{\lambda}_i=\lambda_i=\omega_i$, $\hat{p}_n$ is a consistent estimator of $p$ if $\mleft(\sum_{i=1}^n\lambda_i\mright)^2$ grows quicker than $n$.
\end{corollary}

\balance
\bibliographystyle{ACM-Reference-Format}
\bibliography{references}

\end{document}